\documentclass{birkjour}

\usepackage{graphicx}
\usepackage{dcolumn}
\usepackage{bm}
\usepackage{color}

 \newtheorem{thm}{Theorem}[section]
 \newtheorem{cor}[thm]{Corollary}
 \newtheorem{lem}[thm]{Lemma}
 \newtheorem{prop}[thm]{Proposition}
 \theoremstyle{definition}
 
 \theoremstyle{remark}
 \newtheorem{rem}[thm]{Remark}
 
 \numberwithin{equation}{section}

   \def\R{\mathbb{R}}

   \def\sep{{\,|\ }}

   \def\bz{{z}}

   \def\uep{\bm{e}}

   \def\colon{:}

   \def\gam{\partial \Omega}

\begin{document}

\title[Singular Casimir Elements]{Singular Casimir Elements of the Euler Equation
and Equilibrium Points}

\author[Yoshida]{Zensho Yoshida}

\address{%
Graduate School of Frontier Sciences,
\\
The University of Tokyo,
\\
Kashiwa,
\\
Chiba 277-8561,
\\
Japan
}
\email{yoshida@k.u-tokyo.ac.jp}


\author[Morrison]{Philip J. Morrison}
\address{Department of Physics and Institute for Fusion Studies,
\\
University of Texas,
\\
Austin,
\\
Texas 78712-1060,
\\
USA}
\email{morrison@physics.utexas.edu}

\author[Dobarro]{Fernando Dobarro}
\address{Dipartimento di Matematica e Informatica,
\\
Universit\`a degli Studi di Trieste,
\\
Via Valerio 12/b,
\\
Trieste 34127,
\\
Italy }
\email{fdob07@gmail.com}



\date{\today}



\begin{abstract}
The problem of the nonequivalence of the sets of equilibrium points and energy-Casimir extremal points, which occurs in the noncanonical Hamiltonian formulation of equations describing ideal fluid and plasma dynamics, is addressed in the context of
the Euler equation for an incompressible inviscid fluid.  The problem is traced to a Casimir deficit, where Casimir elements constitute the center of the  Poisson algebra underlying the Hamiltonian formulation, and this leads to a study of  singularities of the Poisson operator defining the Poisson bracket.  The kernel of the Poisson operator, for this typical example of an
infinite-dimensional Hamiltonian system for media in terms of Eulerian variables, is analyzed.  For two-dimensional flows, a rigorously solvable system is formulated.  The nonlinearity of the Euler equation makes the Poisson operator inhomogeneous
on phase space (the function space of the state variable), and it is seen that this creates a singularity where the nullity of the Poisson operator (the ``dimension'' of the center) changes.   The problem is an infinite-dimension generalization of the theory of singular differential equations.
Singular Casimir elements stemming from this singularity are unearthed using a generalization of the functional derivative that occurs in the Poisson bracket.

\end{abstract}

\maketitle


\section{Introduction}
\label{sec:introduction}

Equations that describe ideal fluid and plasma dynamics in terms of Eulerian variables
are Hamiltonian in terms of noncanonical Poisson brackets, degenerate brackets in noncanonical coordinates.  Because of degeneracy, such Poisson brackets possess Casimir elements, invariants that have been used to construct variational principles for equilibria and stability.\footnote{
The first clear usage of the energy-Casimir method for stability appears to be
Kruskal and Oberman\,\cite{Kruskal-Oberman}.  See \cite{morrison98} for an historical discussion. }
However, early on it was recognized that typically there are not enough Casimir elements to obtain all equilibria as extremal points of these variational principles.
In \cite{morrison98,vivek} it was noted that this Casimir deficit is attributable to rank changing of the operator that defines the noncanonical Poisson bracket.
Thus, a mathematical study of the kernel of this operator is indicated,
\footnote{
Also, in \cite{morrison98} it is described how one can for general rank changing cosymplectic operators use a particular kind of constrained variation, called dynamically accessible variations in a sequence of papers beginning with Morrison and Pfirsch\,\cite{Morrison-Pfirsch1989}, but this skirts the central mathematical problem, which is addressed in the present paper.}
and this is the main purpose of the present article.

Recognizing a Hamiltonian flow as a differential operator, the point where the rank of
the Poisson bracket changes is a singularity, from which singular (or intrinsic) solutions stem.
When we consider a Hamiltonian flow on a function space, the problem is an infinite-dimension
generalization of the theory of singular differential equations; the derivatives are functional derivatives, and the construction of singular Casimir elements amounts to  integration in an infinite-dimension space.
In order to facilitate this study,  it is necessary to place the noncanonical Hamiltonian formalism on a more rigorous footing, and this subsidiary purpose is addressed in the context of Euler's equation of fluid mechanics, although the ideas presented are of more general applicability than this particular example.

We start by reviewing finite-dimensional canonical and noncanonical Hamiltonian mechanics, in order to formulate our problem. These dynamical system have the form
\begin{equation}
\frac{d \bz }{dt}= J\, \partial_{\bz} H(\bz)\,,
\label{Hamilton_eq_1}
\end{equation}
where $\bz=(z^1,z^2,\dots, z^m)$ denotes a set of phase space coordinates, $H$ is the Hamiltonian function with $\partial_z$ its gradient,
and the $m\times m$ matrix $J$ (variously called e.g.\ the cosymplectic form, Poisson tensor, or symplectic operator)
is the essence of the Poisson bracket and determines
important Lie algebraic properties\,\cite{morrison98}
(see also Remark\,\ref{remark:Lie-Poisson}).

For canonical Hamiltonian systems of dimension $m=2n$ the matrix $J$ has the form
\begin{equation}
 J_c = \left( \begin{array}{cc}
  ~~0_n \hfill & I_n \\
  -I_n & 0_n\hfill \\
  \end{array} \right)
  \,.
 \label{canHam}
\end{equation}
\emph{Noncanonical} Hamiltonian systems allow a $z$-dependent $J(\bz)$ (assumed here to be a holomorphic function)
to have a kernel, i.e.\  $\mathrm{Rank}(J(\bz))$ may be less than $m$ and may change as a function of $\bz$.

{}From (\ref{Hamilton_eq_1}) it is evident that equilibrium points of the dynamics,
i.e.\ points for which $d \bz /dt =0\ \forall t$, satisfy
\begin{equation}
\partial_{\bz}H(\bz)=0\,.
\label{vacuum_stationary_point}
\end{equation}
However, in the noncanonical case these may not be the only equilibrium points of a given Hamiltonian $H(\bz)$, because  degeneracy gives rise to \emph{Casimir elements} $C(\bz)$, nontrivial (nonconstant) solutions to the differential equation
\begin{equation}
J(\bz) \partial_{\bz} C(\bz) =0\,.
\label{Casmir-1}
\end{equation}
Given such a $C(\bz)$, replacement of the Hamiltonian by $H(\bz)+C(\bz)$
does not change the dynamics.  Thus, an extremal point of
\begin{equation}
\partial_{\bz}[H(\bz)+C(\bz)]=0
\label{stationary_points}
\end{equation}
will also give an equilibrium point.  Note, in light of the homogeneity of (\ref{Casmir-1}) an arbitrary multiplicative constant can be absorbed into $C$ and so (\ref{stationary_points}) can give rise to families of equilibrium points.

If $\mathrm{Rank}(J(\bz))=m=2n$, (\ref{Casmir-1}) has only the trivial solution $C(\bz)=$ constant.
If $\mathrm{Rank}(J(\bz))=2n<m$ and $n$ is constant, then (\ref{Casmir-1})
has $m-2n$ functionally independent solutions (Lie-Darboux theorem).
The problem becomes more interesting
if there is a \emph{singularity} where $\mathrm{Rank}(J(\bz))$ changes:
in this case we have a singular (hyperfunction) Casimir element
(see Fig.\,\ref{fig:foliation}).
For example, consider the one-dimensional system where $J=ix$ ($x\in \mathbb{R}$).
At $x=0$ $\mathrm{Rank}(J)$ drops to 0, and this point is a singular point
of the differential equation $J(x) \partial_x C=0$.
The singular Casimir element is $C(x)={Y}(x)$, where $Y$ is the Heaviside step function.

\begin{figure}[bt]
~~
\begin{center}
  \includegraphics[width=0.8\textwidth]{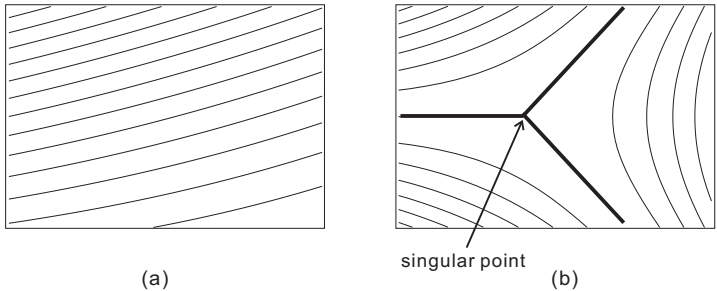}
\end{center}
\caption{
Low-dimensional cartoon  of a foliated phase space.
(a)  Depiction of level-sets  (leaves) of Casimir elements foliating  the phase space.
Since a Casimir element is a constant of motion, every orbit is constrained to a leaf
determined by a Casimir element.  In this cartoon of a two-dimensional phase space, each
Casimir (symplectic) leaf has codimension one, hence in this depiction the
effective space of dynamics has dimension one.
(b) Depiction of phase space with a singular point where $\mathrm{Rank}(J)$ changes.  The
codimension of the Casimir foliation changes, resulting in a
singular Casimir leaf (with the determining Casimir element being a \emph{hyperfunction}).
In the figure the codimension of the Casimir leaf at the singular point is two,
hence the singular point is an equilibrium point.  In higher-dimensional
(infinite-dimensional) phase space, a singular point may have far richer structure.
In Sec.~\ref{sec:Casimirs}, we  delineate how a singular point is
created in an infinite-dimensional phase space by examining an
infinite-dimensional Hamiltonian system of Eulerian fluid (to be formulated
in Sec.~\ref{sec:Euler_equation}).
In Sec.~\ref{sec:equilibrium}, we will study the structure of the singular point
(which is still  infinite-dimensional) by examining  equilibrium points.}
\label{fig:foliation}
\end{figure}

We generalize (\ref{Hamilton_eq_1}) further to include infinite-dimension systems.
Let $u\in X$ be the state variable, where for now $X$ is some unspecified function space,
$\mathcal{J}(u)$ be a linear antisymmetric operator in $X$ that generally depends on $u$ (for a fixed $u$,
$\mathcal{J}(u)$ may be regarded as a linear operator $X\rightarrow X$ --
see Remark \ref{remark:J_operator} below),
and $H(u)$ be a functional $X\rightarrow\mathbb{R}$.
Introducing an appropriate functional derivative (gradient) $\partial_u H(u)$,
we consider  evolution equations of the form
\begin{equation}
\partial_t u = \mathcal{J}(u) \partial_{u} H(u)\,,
\label{Hamilton_eq_2}
\end{equation}
where $\partial_t u:=\partial u/\partial t$.
A Casimir element $C(u)$ (a functional $X\rightarrow\mathbb{R}$) is a nontrivial solution to
\begin{equation}
\mathcal{J} (u) \partial_{u} C(u) =0\,.
\label{Casmir-2}
\end{equation}
We may solve (\ref{Casmir-2}) by two steps:
\begin{enumerate}
\item
Find the kernel of $\mathcal{J}(u)$, i.e., solve a ``linear equation''
(cf. Remark \ref{remark:J_operator})
\begin{equation}
\mathcal{J}(u) v = 0
\label{Casmir-2-1}
\end{equation}
to determine $v$ for a given $u$, which we write as $v(u)$.

\item
``Integrate'' $v(u)$ with respect to $u$ to find a functional $C(u)$
such that $v(u) = \partial_u C(u)$.
\end{enumerate}

As evident in the above finite-dimension example, \emph{step-1} should
involve ``singular solutions'' if $\mathcal{J}(u)$ has \emph{singularities}.
Then, \emph{step-2} will be rather nontrivial -- for \emph{singular Casimir
elements}, we will need to generalize the notion of \emph{functional derivative}.
As mentioned above, the present paper is devoted to such an extension of the notion of Casimir elements
in infinite-dimensional noncanonical Hamiltonian systems. Specifically, we invoke the Euler equation of ideal fluid mechanics as an example, but much carries over to other systems since many fluid and plasma systems share the same operator $ \mathcal{J}$.  In Sec.~\ref{sec:Euler_equation}, we will describe the Hamiltonian form of Euler's equation, which places on a more rigorous footing the formal calculations of \cite{MG80,morrison80,morrison81a,morrison81b,morrison82,olver82}.
In Sec.~\ref{subsec:kernel}, we will analyze the kernel of the corresponding
Poisson operator $\mathcal{J}(u)$ and its singularity.
A singular Casimir element and its appropriate generalized functional derivative (gradient)
will be given in Sec.~\ref{subsec:Casimirs}.

The relation between the (generalized) Casimir elements and equilibrium points
(stationary ideal flows) will be discussed in Sec.~\ref{sec:equilibrium}.
Generalizing (\ref{stationary_points}) to an infinite-dimensional space,
we may find an extended set of equilibrium points by solving
\begin{equation}
\partial_{u}[H(u)+ C(u)]=0.
\label{stationary_points2}
\end{equation}
We note, however, it is still uncertain whether or not every equilibrium point
can be obtained from Casimir elements in this way.
For example, let us consider a simple Hamiltonian $H(u)=\|u\|^2/2$
(in Appendix~A, the Hamiltonian of the Euler equation is given in terms of the velocity field $\bm{u}$,
which here corresponds to the state variable $u$).
Then, $\partial_{u} H(u)=u$ and (\ref{Hamilton_eq_2}) reads
\begin{equation}
\partial_t u = \mathcal{J}(u) u .
\label{Hamilton_eq_3}
\end{equation}
The totality of nontrivial equilibrium points is $\mathrm{Ker}(\mathcal{J}(u))$.
For $v\in\mathrm{Ker}(\mathcal{J}(u))$ to be characterized by
(\ref{stationary_points2}), which now simplifies to $u=- \partial_u C(u)$,
we encounter the ``integration problem,'' i.e., we have to
construct $C(u)$ such that $v(u) = \partial_u C(u)$
for every $v(u)\in\mathrm{Ker}(\mathcal{J}(u))$
-- this may not be always possible.
On the other hand, even for a given $C(u)$, the ``nonlinear equation''
$u=- \partial_u C(u)$ does not necessarily have a solution --
in Sec.\,\ref{subsec:no-solution} we will show some examples of no-solution equations.
While we leave this question (the ``integrability'' of all equilibrium points)
open, the present effort shows it is sometimes possible and provides a more complete understanding of the stationary states of infinite-dimensional dynamical systems.  In Sec.~\ref{sec:conclusion} we give some concluding remarks.

\begin{rem}
\label{remark:Lie-Poisson}
We endow the phase space $X$ of state variable $u$ ($\bz$ if $X$ is finite-dimensional)
with an inner product $(a, b)$.
Let $F :\, X\rightarrow \mathbb{R}$ be an arbitrary smooth functional (function if $X$ is finite-dimensional).
If $u$ obeys (\ref{Hamilton_eq_2}) in a Hilbert space (or (\ref{Hamilton_eq_1}) for finite-dimensional systems),
the evolution of $F(u)$ obeys
$d F({u})/dt = [F({u}), H({u})]$,
where
\[
[F({u}), H({u})] = (\partial_{{u}} F({u}),\mathcal{J}({{u}}) \partial_{{u}} H({u}))
\]
is an antisymmetric bilinear form.
If this bracket $[~,~]$ satisfies the Jacob identity, it defines a Poisson algebra, a Lie algebra realization on functionals.
A Casimir element $C$ is a member of the \emph{center} of the Poisson algebra,
i.e., $[C, G]=0$ for all $G$.
The bracket defined by the Poisson operator of Sec.\,\ref{subsec:Poisson}
satisfies the Jacobi's identity \cite{morrison81a,morrison82,morrison98}.  The Jacobi identity is satisfied for all Lie-Poisson brackets, a class of Poisson  brackets  that describe matter that are built from the structure constants of Lie algebras (see, e.g., \cite{morrison98}).  For finite-dimensional systems there is a beautiful geometric interpretation of such brackets where phase space is the dual of the Lie algebra and surfaces of constant Casimirs, coadjoint orbits, are symplectic manifolds.   Unfortunately, in infinite-dimensions, i.e., for nonlinear partial differential equations,  there are functional analysis challenges that limit this interpretation  (see, e.g., \cite{khesin}).  (For example, for the incompressible Euler fluid equations the group is that of volume preserving diffeomorphisms.)  In terms of this interpretation, the analysis of the present paper can be viewed as a local probing of the coadjoint orbit.
\end{rem}

\begin{rem}
\label{remark:J_operator}
In (\ref{Hamilton_eq_2}),
the operator $\mathcal{J}(u)$ must be evaluated at the common $u$
of $\partial_{u}H(u)$, thus $\mathcal{J}(u)\partial_{u}H(u)$ is a
nonlinear operator with respect to $u$.
However, the application of $\mathcal{J}(u)$ (or $\mathcal{J}(u)\partial_{u}$)
may be regarded as a linear operator in the sense that
\[
\mathcal{J}(u)(a v + b w) = a\mathcal{J}(u)v + b \mathcal{J}(u)w
\]
or
\[
\mathcal{J}(u)\partial_{u} [aF(u)+bG(u)]
=a\mathcal{J}(u)\partial_{u} F(u)+b\mathcal{J}(u)\partial_{u} G(u).
\]
Note that $\mathcal{J}(u) v$ (for $v=\partial_{u} F(u)$) is \emph{not}
$\mathcal{J}(v)v$.
\end{rem}

\section{Hamiltonian Form of the Euler Equation}
\label{sec:Euler_equation}

Investigation of the Hamiltonian form of ideal fluid mechanics has a long history.
Its essence is contained in Lagrange's original work \cite{lagrange} that described the fluid in terms of the ``Lagrangian'' displacement. Important subsequent contributions are due to Clebsch\,\cite{clebsch1,clebsch2} and Kirchhoff \cite{kirch}.
In more recent times, the formalism has been addressed in various ways by many authors (e.g.\,\cite{eckart,newcomb,arnold1,arnold2,arnold3,zakharov}).
Here we follow the noncanonical Poisson bracket description as described in \cite{MG80,morrison98}.
Analysis of the kernel of $\mathcal{J}$ requires careful definitions.  For this reason we review rigorous results about Euler's equation in Sec.~\ref{subsec:vorticity_equation}, followed by explication of various aspects of the Hamiltonian description in
Secs.~\ref{subsec:Hamiltonian} 
-- \ref{subsec:Hamiltonian-form}.
This places aspects of the noncanonical Poisson bracket formalism of\,\cite{morrison80,morrison81a,morrison82,olver82} on a more rigorous footing;
of particular interest, of course, is the Poisson operator $\mathcal{J}$ that defines the Poisson bracket.

\subsection{Vorticity Equation}
\label{subsec:vorticity_equation}
Euler's equation of motion for an incompressible inviscid fluid is
\begin{eqnarray}
\partial_t \bm{u} + (\bm{u}\cdot\nabla)\bm{u} = -\nabla p
~~(\mathrm{in~}\Omega)\,,
\label{Euler-1}
\\
\nabla\cdot\bm{u}=0
~~(\mathrm{in~}\Omega)\,,
\label{incompressibility}
\\
\bm{n}\cdot\bm{u}=0
~~(\mathrm{on~}\gam )\,,
\label{BC}
\end{eqnarray}
where $\Omega$ is a bounded domain in $\R^n$ ($n= 2$ or 3) with
a sufficiently smooth ($C^{2+\epsilon}$-class) boundary $\gam$, $\bm{n}$ is the unit vector normal to $\gam$,
$\bm{u}$ is an $n$-dimensional vector field representing the velocity field,
and $p$ is a scalar field representing the fluid pressure (or specific enthalpy);
all fields are real-valued functions of time $t$ and the spatial coordinate $\bm{x}\in\Omega$.

We may rewrite (\ref{Euler-1}) as
\begin{equation}
\partial_t \bm{u} = \bm{u}\times\bm{\omega} -\nabla \tilde{p}
\,,
\label{Euler-2}
\end{equation}
where $\bm{\omega}=\nabla\times\bm{u}$ is the vorticity and $\tilde{p}=p+u^2/2$
is the total specific energy.
The curl derivative of (\ref{Euler-2}) gives the \emph{vorticity equation}
\begin{equation}
\partial_t \bm{\omega} = \nabla\times(\bm{u}\times\bm{\omega})
\,.
\label{Vorticity_eq}
\end{equation}

We prepare basic function spaces pertinent to the mathematical formulation of the
Euler equation. Let $L^2(\Omega)$ be the Hilbert space of Lebesgue-measurable and square-integrable
real vector functions on $\Omega$, which is
endowed with the standard inner product $(\bm{a},\bm{b})=\int_\Omega \! dx\, \bm{a}\cdot\bm{b}$
and  norm $\| \bm{a} \|=(\bm{a},\bm{a})^{1/2}$.
We will also use the standard notation for Sobolev spaces (for example, see\,\cite{Brezis2}).
We define
\begin{equation}
L_\sigma^2(\Omega) = \{ \bm{u}\in L^2(\Omega)\sep \nabla\cdot\bm{u}=0,~\bm{n}\cdot\bm{u}=0\},
\label{L^2_sigma}
\end{equation}
where $\bm{n}\cdot\bm{u}$ denotes the trace of the normal component of $\bm{u}$ onto the
boundary $\gam $, which is a continuous map from $\{ \bm{u}\in L^2(\Omega)\sep
\nabla\cdot\bm{u}\in L^2(\Omega)\} $ to $H^{-1/2}(\gam )$.
We have an orthogonal decomposition
\begin{equation}
L^2(\Omega)=L_\sigma^2(\Omega) \oplus \{\nabla\theta\sep \theta\in H^1(\Omega)\}.
\label{decomposition_of_L2}
\end{equation}
Every $\bm{u}\in L_\sigma^2(\Omega)$ satisfies the conditions (\ref{incompressibility})
and (\ref{BC}), thus we will consider (\ref{Euler-2}) to be an evolution equation in the
function space $L_\sigma^2(\Omega)$ (cf.\  Appendix A).

Hereafter, we assume that the spatial domain has dimension $n=2$, and
$\Omega\subset\R^2$ is a smoothly bounded and simply connected (genus=0) region.
\footnote{
Generalization to a multiply connected region is not difficult and is done as follows: first
extract $L^2_H(\Omega)=\{\bm{u}\in L^2_\sigma(\Omega)\sep \nabla\times\bm{u}=0\}$
from $L^2_\sigma(\Omega)$, i.e.\  express $L^2_\sigma(\Omega)=L^2_H(\Omega)\oplus L^2_\Sigma(\Omega)$,
where the dimension of the subspace $L^2_H(\Omega)$ is equal to the genus of $\Omega$.
Then, the projection of $\bm{u}$, which obeys (\ref{Euler-1})-(\ref{BC}), is shown to be constant throughout the evolution, whence we may regard (\ref{Euler-2}) as an evolution equation in $L^2_\Sigma(\Omega)$.}
For convenience in formulating equations, we immerse $\Omega\subset\R^2$ in $\R^3$ by adding a
``perpendicular'' coordinate $z$, and we write $\uep=\nabla z$.

\begin{lem}
\label{lemma:Clebsch_representation}
Every two-dimensional vector field $\bm{u}$ satisfying
the incompressibility condition (\ref{incompressibility}) and
the vanishing normal boundary condition (\ref{BC})
can be written as
\begin{equation}
\bm{u}=\nabla\varphi\times\uep
\label{Clebsch_2form}
\end{equation}
with a single-value function $\varphi$ such that $\varphi|_{\gam} =0$,
i.e.,
\begin{equation}
L_\sigma^2(\Omega) = \{ \nabla\varphi\times\uep \sep
\varphi\in H^1_0(\Omega) \}.
\label{Clebsch_representation}
\end{equation}
\end{lem}

\begin{proof}
For the convenience of the reader, we sketch the proof of this
frequently-used lemma.\footnote{
The function $\varphi$ is sometimes called a Clebsch potential.
To represent an incompressible flow of dimension $n$, we need $n-1$ \emph{Clebsch potentials}
$\varphi_1,\cdots,\varphi_{n-1}$,
where each $\varphi_j$ does not have a uniquely determined boundary condition\,\cite{Yoshida2009}.
Hence,  the vorticity representation is not effective in higher dimensions.
In Appendix A, we invoke another method to eliminate the pressure term and formulate
the problem in an alternative form, which applies in general space dimension.
See e.g.\ \cite{eckart,clebsch1,clebsch2,morrison82,morrison06}
for discussions of various potential representations.
}
Evidently, $\nabla\cdot( \nabla\varphi\times\uep)=\nabla\cdot[\nabla\times(\varphi\uep)]=0$,
and $\left.\bm{n}\cdot( \nabla\varphi\times\uep)\right|_{\gam } = \left.(\uep\times\bm{n})\cdot\nabla\varphi\right|_{\gam }=0$
if $\varphi\in H^1_0(\Omega)$.
Thus, the linear space $X=\{ \nabla\varphi\times\uep\sep \varphi\in H^1_0(\Omega)\}$
is contained in $L_\sigma^2(\Omega)$.
And, the orthogonal complement of $X$ in $L^2_\sigma(\Omega)$ contains only the zero vector:
Suppose that $\bm{u}\in L^2_\sigma(\Omega)$ satisfies
\begin{equation}
(\bm{u}, \nabla\varphi\times\uep)=0 \quad \forall\varphi\in H^1_0(\Omega).
\label{Clebsch-1}
\end{equation}
By the generalized Stokes formula, we find
$(\bm{u}, \nabla\varphi\times\uep)=(\uep\cdot\nabla\times\bm{u}, \varphi)$.
Since $\nabla\times\bm{u}$ has only the $\uep$ component, (\ref{Clebsch-1})
implies $\nabla\times\bm{u}=0$.  Since $\bm{u}\in L^2_\sigma(\Omega)$, we also have
$\nabla\cdot\bm{u}=0$ and $\left.\bm{n}\cdot\bm{u}\right|_{\gam }=0$.  In a simply connected $\Omega$,
the only such $\bm{u}$ is the zero vector.  Hence, we have (\ref{Clebsch_representation}).
\end{proof}

Using the representation (\ref{Clebsch_2form}), we may formally calculate
\[
\bm{\omega}=\nabla\times\bm{u}=(-\Delta \varphi)\uep=: \omega\uep\,.
\]
The vorticity equation (\ref{Vorticity_eq}) simplifies to a single $\uep$-component equation:
\footnote{
For sufficiently smooth $\bm{u}$, the two-dimensional vorticity equation (\ref{Vorticity_eq-2D})
has the form of a nonlinear Liouville equation. The corresponding Hamiltonian equations (characteristic ODEs)
are given in terms of a streamfunction $\varphi$ as
$d\bm{x}/dt = (\partial_y \varphi, -\partial_x\varphi ) = \bm{u}$.
By the boundary condition
$\bm{n}\cdot\bm{u}=(\uep\times\bm{n})\cdot\nabla(\mathcal{K}\omega)=0$,
the characteristic curves are confined in $\Omega$.  Hence, we do not need
(or, cannot impose) a boundary condition on $\omega$, and the single equation (\ref{Vorticity_eq-2D}) determines
the evolution of $\omega$ and the velocity field is obtained from
$\bm{u}=\nabla(\mathcal{K}\omega)\times\uep$. }
\begin{equation}
\partial_t {\omega} = \{\omega , \mathcal{K} \omega  \} ~~(\mathrm{in~}\Omega),
\label{Vorticity_eq-2D}
\end{equation}
where
\[
\{a,b\}=-\nabla a \times \nabla b \cdot \uep
= \partial_y a \cdot \partial_x b - \partial_x a \cdot \partial_y b,
\]
and $\mathcal{K}$ is the inverse map of $-\Delta$ with the Dirichlet boundary condition,
i.e., $\mathcal{K}\colon\omega\mapsto \varphi$ gives the solution of the Laplace equation
\begin{equation}
-\Delta \varphi = \omega~(\mathrm{in~}\Omega),
\quad \varphi=0~(\mathrm{on~}\gam ).
\label{operator-K}
\end{equation}
As is well-known, $\mathcal{K}\colon L^2(\Omega)\rightarrow H^1_0(\Omega)\cap H^2(\Omega)$
is a self-adjoint compact operator.
For $\varphi\in H^1_0(\Omega)$, we define $\omega = -\Delta \varphi$ as a member of
$H^{-1}(\Omega)$, the dual space of $H^1_0(\Omega)$ with respect to the inner-product of $L^2(\Omega)$.
The inverse map (weak solution), then, defines $\mathcal{K}\colon H^{-1}(\Omega)\rightarrow H^1_0(\Omega)$.

\begin{lem}
\label{lemma:weak}
We regard the vorticity equation (\ref{Vorticity_eq-2D}) as an evolution equation
in $H^{-1}(\Omega)$, i.e., we consider the weak form:
\begin{equation}
(\partial_t\omega-\{\omega,\mathcal{K}\omega\},\phi)=0 \quad \forall\phi\in H^1_0(\Omega)\,.
\label{Vorticity_eq-2D-math}
\end{equation}
By the relations $\varphi=\mathcal{K}\omega$, $\bm{u}=\nabla\varphi\times\uep$, and
$\bm{\omega}=\omega\uep$,
(\ref{Vorticity_eq-2D-math})
is equivalent to the Euler equation (\ref{Euler-2})
as an evolution equation in $L^2_\sigma(\Omega)$.
\end{lem}

\begin{proof}
In the topology of $L^2_\sigma(\Omega)$, the Euler equation (\ref{Euler-2}) reads as
\begin{equation}
(\partial_t \bm{u} - \bm{u}\times\bm{\omega} +\nabla \tilde{p}, \bm{v})=0
\quad \forall \bm{v}\in L^2_\sigma(\Omega).
\label{Euler-2'}
\end{equation}
By (\ref{decomposition_of_L2}), the left-hand side of (\ref{Euler-2'}) reduces to
$(\partial_t \bm{u} - \bm{u}\times\bm{\omega}, \bm{v})$.
By Lemma\,\ref{lemma:Clebsch_representation}, we may put
$\bm{v}=\nabla\phi\times\uep=\nabla\times(\phi\uep)$ with
$\phi\in H^1_0(\Omega)$.
Finally, plugging this representation into (\ref{Euler-2'}), we obtain
\begin{eqnarray*}
(\partial_t \bm{u} - \bm{u}\times\bm{\omega}, \nabla\times(\phi\uep))
&=&
(\uep\cdot\nabla\times(\partial_t \bm{u} - \bm{u}\times\bm{\omega}),\phi)
\\
&=&
(\partial_t\omega-\{\omega,\varphi \},\phi).
\end{eqnarray*}
Hence, (\ref{Euler-2'}) is equivalent to (\ref{Vorticity_eq-2D-math}).
\end{proof}

\subsection{Hamiltonian}
\label{subsec:Hamiltonian}

Now we consider the Hamiltonian form of the vorticity equation (\ref{Vorticity_eq-2D}) -- to be precise, its
``weak form'' (\ref{Vorticity_eq-2D-math}) (cf.\   Appendix A which treats the Euler equation of (\ref{Euler-2})).

 First we note that the natural choice for the Hamiltonian is $H=\|\bm{u}\|^2/2$, the ``energy'' of the flow $\bm{u}$.
Using $\bm{u}=\nabla\varphi\times\uep$, we may rewrite $H=\|\nabla\varphi\|^2/2 = (\varphi, -\Delta\varphi)/2$.
Selecting the vorticity $\omega$ as the state variable, we define (by relating $\varphi=\mathcal{K}\omega$)
\begin{equation}
H(\omega) = \frac{1}{2} \int_\Omega\!dx\, (\mathcal{K}\omega)\cdot\omega\,,
\label{energy_by_omega}
\end{equation}
which is a continuous functional on $H^{-1}(\Omega)$.  This is equivalent to the square of the norm of $H^{-1}(\Omega)$, i.e.,
the \emph{negative norm} induced by $H^1_0(\Omega)$.

\subsection{Gradient in Hilbert Space}
\label{subsec:gradient}

Next, we consider the gradient of a functional in function space.
Let $\Phi(u)$ be a functional defined on a Hilbert space $X$.
A small perturbation $\epsilon\tilde{u}\in X$ ($|\epsilon|\ll 1$, $\|\tilde{u}\|=1$)
will induce a variation $\delta\Phi(u;\tilde{u})=\Phi(u+\epsilon\tilde{u})-\Phi(u)$.
If there exists a $g\in X^*=X$ such that
$\delta\Phi(u;\tilde{u})=\epsilon(g,\tilde{u}) + O(\epsilon^2)$ for every $\tilde{u}$,
then we define $\partial_u\Phi(u)=g$, and call it the gradient of $\Phi(u)$.
Evidently, the variation $|\delta\Phi(u;\tilde{u})|$ is maximized, at each $u$, by
$\tilde{u}=\partial_u\Phi(u)/\|\partial_u\Phi(u)\|$.
The notion of gradient will be extended for a class of
``rugged'' functionals, which will be used to define singular Casimir elements in Sec.\,\ref{subsec:Casimirs}.
As for the Hamiltonian, however, we may assume it to be a smooth functional.  The pertinent Hilbert space is $L^2(\Omega)$, on which
the Hamiltonian $H(\omega)$ is differentiable;
using the self-adjointness of $\mathcal{K}$, we obtain
\begin{equation}
\partial_\omega H(\omega) = \mathcal{K}\omega.
\label{delH}
\end{equation}
Note that the gradient $\partial_\omega H(\omega)$ may be evaluated for
every $\omega\in H^{-1}(\Omega)$ with the value in $H^1_0(\Omega)$.

\subsection{Noncanonical Poisson Operator}
\label{subsec:Poisson}

Finally, we describe the noncanonical Poisson operator $\mathcal{J}(\omega)$ of \cite{morrison80,morrison81a,morrison82,olver82}.
Formally, we have
\begin{equation}
\mathcal{J}(\omega) \psi = \left[(\partial_y \omega)\partial_x
- (\partial_x \omega)\partial_y \right] \psi = \{ \omega, \psi \}\,,
\label{Vorticity_eq-J}
\end{equation}
which indicates $\omega$ must be a ``differentiable'' function.
However, we will need to reduce this regularity requirement on $\omega$.
Thus, we turn to the \emph{weak formulation} that is amenable to the interpretation
of the evolution in $H^{-1}(\Omega)$ (see Lemma\,\ref{lemma:weak}).
Formally, we may calculate
\begin{equation}
(\mathcal{J}(\omega) \psi, \phi) = (\{\omega,\psi\},\phi) = (\omega, \{\psi,\phi\})\,,
\label{J-operator-1}
\end{equation}
with the right-hand-side finite (well-defined) for $\omega\in C(\Omega)$ and
$\psi, \phi \in H^1_0(\Omega)$.  In fact,
\begin{eqnarray*}
\left|(\omega, \{\psi,\phi\})\right| &\leq& \|\omega\|_{{\sup}} \int_\Omega \!dx\,   |\{\psi,\phi\}|
\\
&\leq& \|\omega\|_{{\sup}} \int_\Omega\!dx\,   |\nabla\psi|\,|\nabla\phi|
\\
&\leq& \|\omega\|_{{\sup}}\,  \|\nabla\psi\|\,\|\nabla\phi\|,
\end{eqnarray*}
where $\|\omega\|_{\sup}={\sup}_{x\in\Omega} |\omega(x)|$.
Hence, we may consider the right-hand-side of (\ref{J-operator-1})
to be a bounded linear functional of $\phi\in H^1_0(\Omega)$,
with $\omega$ and $\psi$ acting as two parameters.  We denote this by $(\omega, \{\psi,\phi\})=: F(\omega,\psi;\phi)$.
By this functional, we ``define'' $\mathcal{J}(\omega)\psi$
on the left-hand-side of (\ref{J-operator-1})
as a member of $H^1_0(\Omega)^*=H^{-1}(\Omega)$, i.e., we put
\[
(\mathcal{J}(\omega)\psi,\phi) \colon =
F(\omega,\psi;\phi) \quad \forall\phi\in H^1_0(\Omega).
\]

For a given $\omega\in C(\Omega)$, we may consider that $\mathcal{J}(\omega)$ is a bounded
linear map operating on $\psi$, i.e., $\mathcal{J}(\omega)\colon H^1_0(\Omega)\rightarrow H^{-1}(\Omega)$.
Evidently, $(\mathcal{J}(\omega)\psi,\phi)=-(\psi,\mathcal{J}(\omega)\phi)$, i.e.,
$\mathcal{J}(\omega)$ is antisymmetric.

\subsection{Hamiltonian Form of Vorticity Equation}
\label{subsec:Hamiltonian-form}

Combining the above definitions of the Hamiltonian
$H(\omega)$, the gradient $\partial_\omega$, and the
noncanonical Poisson operator $\mathcal{J}(\omega)$, we can
write the vorticity equation (\ref{Vorticity_eq-2D}) in the form
\begin{equation}
\partial_t \omega = \mathcal{J}(\omega) \partial_\omega H(\omega).
\label{Vorticity_eq-2DH}
\end{equation}
As remarked in Lemma\,\ref{lemma:weak}, (\ref{Vorticity_eq-2DH}) is
an evolution equation in $H^{-1}(\Omega)$ (cf.\ Appendix A for the $\bm u$ formulation).

For every fixed $\omega\in C(\Omega)$, $\mathcal{J}(\omega)$ may be regarded as a bounded
linear map of $H^1_0(\Omega)\rightarrow H^{-1}(\Omega)$, where the bound changes as a
function of $\omega$.   And $\partial_\omega H(\omega)$ is a bounded linear map of
$H^{-1}(\Omega)\rightarrow H^1_0(\Omega)$.
Hence, each element composing the right-hand side (generator) of the evolution equation
(\ref{Vorticity_eq-2DH}) is separately regular.
However, their nonlinear combination can create a problem:
As noted in Remark \ref{remark:J_operator}, we must evaluate the operator $\mathcal{J}(\omega)$
at the common $\omega$ of $\partial_\omega H(\omega)$.
While $\partial_\omega H(\omega)$ can be evaluated for every $\omega\in H^{-1}(\Omega)$
with its range $= H^1_0(\Omega)$ =
domain of $\mathcal{J}(\omega)$, \emph{if} $\mathcal{J}(\omega)$ is defined; however,
we can define the operator $\mathcal{J}(\omega)$ only for $\omega\in C(\Omega)$.
The difficulty of this nonlinear system is now delineated by the
singular behavior of the Poisson operator $\mathcal{J}(\omega)$ as a function of $\omega$
-- if the orbit $\omega(t)$ (in the
function space $H^{-1}(\Omega)$) runs away so as to increase $\|\omega\|_{\sup}$,
the evolution equation (\ref{Vorticity_eq-2DH}) will breakdown.

To match the combination of $\mathcal{J}(\omega)$ and $\partial_\omega H(\omega)$,
the domain of the total generator $\mathcal{J}(\omega) \partial_\omega H(\omega)$
must be restricted in $C(\Omega)$.
Fortunately, this domain is not too small; the regular (classical) solutions
for an appropriate initial condition lives in this domain,
i.e., if a sufficiently smooth initial condition is given, the
orbit stays in the region where $\|\omega\|_{\sup}$ is bounded \cite{kato1967}.

\section{Casimir Elements}
\label{sec:Casimirs}

\subsection{The Kernel of $\mathcal{J}(\omega)$}
\label{subsec:kernel}

We begin with a general representation of the kernel of the noncanonical
Poisson operator $\mathcal{J}(\omega)$, which will be a subset of its domain $H^1_0(\Omega)$ (see Sec.~\ref{subsec:Poisson}).

\begin{lem}
\label{lemma:kernel}
For a given $\omega\in C(\Omega)$, $\psi\in H^1_0(\Omega)$
is an element of $\emph{Ker}(\mathcal{J}(\omega))$, iff there is $\theta\in H^1(\Omega)$ such that
\begin{equation}
\omega\nabla\psi = \nabla\theta.
\label{Kernel-general'}
\end{equation}
This implies that
\begin{equation}
\emph{Ker}(\mathcal{J}(\omega)) =
\{ \psi \in H^1_0(\Omega)\sep \omega\nabla\psi \in L^2_\sigma(\Omega)^\perp \}.
\label{kernel-general}
\end{equation}
\end{lem}

\begin{proof}
By the definition (\ref{Vorticity_eq-J}), $\psi \in \mathrm{Ker}(\mathcal{J}(\omega))$ implies
$\{\omega, \psi \}=0$ in the topology of $H^{-1}(\Omega)$, i.e.,
\begin{equation}
(\{\omega, \psi \},\phi) \equiv -(\omega\nabla\psi,\nabla\phi\times\uep)
=0 \quad \forall \phi\in H^1_0(\Omega).
\label{Kernel}
\end{equation}
By Lemma\,\ref{lemma:Clebsch_representation}, (\ref{Kernel}) implies that
\begin{equation}
(\omega\nabla\psi,\bm{v}) =0
\quad \forall \bm{v}\in L^2_\sigma(\Omega).
\label{Kernel''}
\end{equation}
Remembering (\ref{decomposition_of_L2}),
we obtain (\ref{Kernel-general'}) and (\ref{kernel-general}).
\end{proof}

To construct a Casimir element from $\psi\in\mathrm{Ker}(\mathcal{J}(\omega))$,
we will need a more ``explicit'' relation between $\omega$ and $\psi$.  We will show how such a relation is available for a sufficiently regular $\omega$.

Let us start by assuming $\omega\in C^1(\Omega)$.  Then, we may evaluate
$\mathcal{J}(\omega)\psi$ as
$\{\omega,\psi\}\equiv -\nabla\omega\times\nabla\psi\cdot\uep \in L^2(\Omega)$.
Therefore, $\psi\in H^1_0(\Omega)$ belongs to $\mathrm{Ker}(\mathcal{J}(\omega))$, iff
\begin{equation}
\{\omega,\psi\}=0 \quad [\in L^2(\Omega)].
\label{Kernel-general''}
\end{equation}
Equation (\ref{Kernel-general''}) implies that two vectors
$\nabla\omega\in C(\Omega)$ and $\nabla\psi\in L^2(\Omega)$
must align almost everywhere in $\Omega$, excepting any ``region"
in which one of them is zero.
Such a relationship between $\omega$ and $\psi$ can be represented, by invoking
a certain scalar $\zeta(x,y)$, as
\begin{equation}
\omega=f(\zeta), \quad \psi=g(\zeta).
\label{Kernel-general''-solution}
\end{equation}
The simplest solution is given by $\psi=g(\omega)$ (i.e., $f=$ identity).
In later discussion, we shall invoke a nontrivial $f$ to represent a wider class of solutions.

We note that the condition $\psi\in H^1_0(\Omega)$ implies the boundary condition $\psi|_{\gam} =0$.
If $\omega\neq$ constant on some $\gam'\subseteq\gam$, integrating (\ref{Kernel-general''})
with this boundary condition yields $\psi\equiv0$ along every contour of $\omega$ which intersects $\gam'$
[the contours of $\omega$ are the Cauchy characteristics of (\ref{Kernel-general''}), and
$\psi|_{\gam'}=0$ poses a non-characteristic \emph{initial condition}].
We denote by $\Omega_\circ(\omega)$
the largest region in $\Omega$ (not necessarily a connected set)
which is bounded by a level set (contour) of $\omega$.  See Fig.~\ref{fig:schematic}.
If $\omega|_{\gam} \neq$ constant, $\Omega_\circ(\omega)$ is smaller than $\Omega$,
and then, every level set of $\omega$ in $\Omega\setminus\Omega_\circ(\omega)$ intersects
$\gam $.
Hence, $\mathrm{supp}(\psi) :=$ closure of $\{x\in\Omega\sep \psi(x)\neq 0\}
\subseteq\Omega_\circ(\omega)$.
We shall assume that $\Omega_\circ(\omega)\neq\emptyset$
for the existence of nontrivial $\psi\in\mathrm{Ker}(\mathcal{J}(\omega))$.

\begin{figure}[tb]
\begin{center}
  \includegraphics[width=0.5\textwidth]{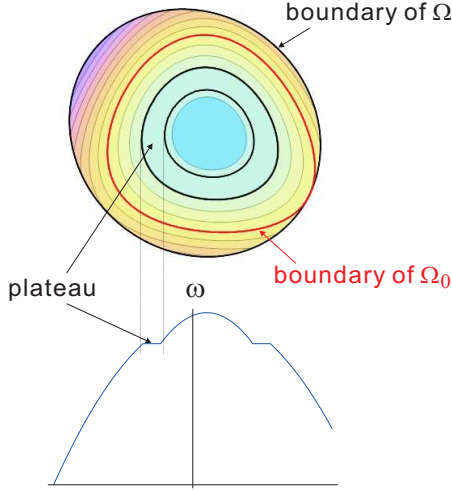}
\end{center}
\caption{Depiction of $\Omega_\circ(\omega)$, the largest subset of $\Omega$ bounded by an $\omega$ level set.  Note, the complement of $\Omega_\circ(\omega)$ contains level sets that intersect $\gam$.  Also depicted is the plateau set $\Omega_p\subset\Omega$, a region where $\omega$=constant.
}
\label{fig:schematic}
\end{figure}

Now we make a moderate generalization about the regularity:
Suppose that $\omega$ is Lipschitz continuous, i.e., $\omega\in C^{0,1}(\Omega)$.
Then, $\omega$ has a classical gradient $\nabla\omega$ almost everywhere in $\Omega$ (see e.g.\ \cite{stein}),
and $|\nabla\omega|$ is bounded.  Note, $\omega$ may fail to have a classical $\nabla\omega$ on a measure-zero subset $\Omega_s\subset\Omega$, but for $\bm{x}\in\Omega_s$ we may define a set-valued \emph{generalized gradient} (see \cite{Clarke1975}).
With a Lipschitz continuous function $g\colon\R\rightarrow\R$
we can solve (\ref{Kernel-general'}) by
\begin{equation}
\psi = g(\omega), \quad \theta=\vartheta(\omega),
\label{Kernel-1}
\end{equation}
where $\vartheta(\xi) = g(\xi)\xi - \int \! d\xi \, g(\xi)$.
To meet the boundary condition $\psi|_{\gam} =0$, $g$ must satisfy
\begin{equation}
g(\omega(\bm{x})) = 0 \quad \forall \bm{x}\in \gam .
\label{Kernel-summary-b}
\end{equation}

However, the solution (\ref{Kernel-1}) omits a different type of solution that emerges with
a \emph{singularity} of $\mathcal{J}(\omega)$: If $\omega$ has a ``plateau,''
i.e., $\omega=\omega_0$ (constant) in a finite
region $\Omega_p\subseteq\Omega$ (see Fig.~\ref{fig:schematic}), the operator $\mathcal{J}(\omega)$
trivializes as $\mathcal{J}(\omega)=\{\omega_0, \cdot\}=0$ in $\Omega_p$
(i.e., the ``rank'' drops to zero; recall the example of Sec.~\ref{sec:introduction}),
and within $\Omega_p$ we can solve (\ref{Kernel-general'}) by
an arbitrary $\psi$ with $\theta=\omega_0\psi$.
Notice that the solution (\ref{Kernel-1}) restricts
$\psi=g(\omega_0)$ = constant in $\Omega_p$.
To remove this degeneracy, we abandon the continuity of $g$ and, for simplicity,
assume that $\omega$ has only a single plateau.
First we invoke the reversed form [cf.\  (\ref{Kernel-general''-solution})]:
\begin{equation}
\omega = f(\psi)\,,
\label{Kernel-2}
\end{equation}
where we assume that $f$ is a Lipschitz continuous monotonic function.
Denoting $\vartheta(\eta)=\int\! d\eta\, f(\eta)$,
we may write $\omega\nabla\psi=\nabla \vartheta(\psi)$, where the gradients on both sides evaluate classically almost everywhere in $\Omega$ if $\psi$ is Lipschitz continuous.
If the function $f(\psi)$ is flat on some interval, $f(\psi)=\omega_0$ = constant for $\psi_- < \psi < \psi_+$,
a plateau appears in the distribution of $\omega$.  See Figs.~\ref{fig:schematic} and \ref{fig:plateau}(a).
Since the present mission is to find $\psi$ for a given $\omega$,
we transform (\ref{Kernel-2}) back to (\ref{Kernel-1}) with the definition  $g=f^{-1}$.  (In Sec.~\ref{sec:equilibrium}, however, we will seek an equilibrium $\omega$ that is characterized by a Casimir element, and then, the form (\ref{Kernel-2}) will be
invoked again). A plateau in the graph of $f$ will, then, appear as a ``jump''
in the graph of $g$. See Fig.~\ref{fig:plateau}(b).

\begin{figure}[tb]
\begin{center}
  \includegraphics[width=0.8\textwidth]{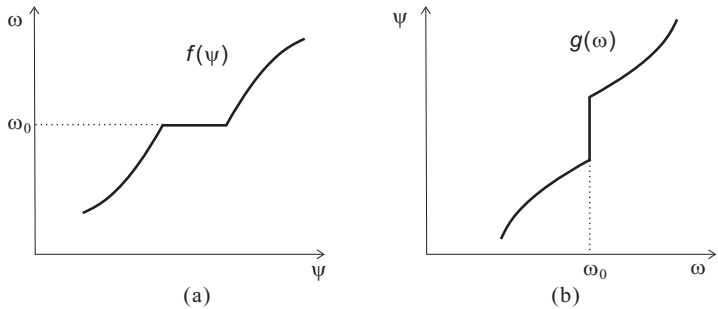}
\end{center}
\caption{Depiction of the dual relationships between $\psi$ and $\omega$.
(a) The relation $\omega=f(\psi)$ with the plateau singularity.
(b) The singular (discontinuous) function $\psi=g(\omega)$ illustrating the kernel element
that stems from a plateau of $\omega$.
}
\label{fig:plateau}
\end{figure}

We now allow the function $g(\omega)$ to have a jump at $\omega_0=\omega|_{\Omega_p}$.
Formally we write $g(\omega)=g_L(\omega) + \alpha{Y}(\omega-\omega_0)$, with  $g_L(\omega)$
a Lipschitz continuous function, ${Y}(\omega-\omega_0)$ the step function,
 and $\alpha$ a constant determining the width of the jump.
We connect the graph of the step function by filling the gap between
$\lim_{\omega\uparrow\omega_0}g(\omega)=\psi_-$ and
$\lim_{\omega\downarrow\omega_0}g(\omega)=\psi_+ = \psi_-+\alpha$. See Fig.~\ref{fig:plateau}(b).
Since $g(\omega)$ is multi-valued at $\omega=\omega_0$, $\psi(\bm{x})=g(\omega(\bm{x}))$
is arbitrary in the range of $[\psi_-,\psi_+]$, i.e.\ in $\Omega_p$.
Choosing a sufficiently smooth $\psi$ in $\Omega_p$, we may assume $\psi\in H^1(\Omega)$.

Summarizing the above discussions (and making an obvious generalization), we have

\begin{thm}
\label{theorem:kernel2}
Suppose that $\omega\in C^{0,1}(\Omega)$ and $\Omega_\circ(\omega)\neq\emptyset$.
Then $\mathrm{Ker}(\mathcal{J}(\omega))$ contains nontrivial elements, and
a part of them can be represented as
\begin{equation}
\psi = g(\omega),
\label{Kernel-summary}
\end{equation}
where $g(\xi)$ is an arbitrary function that satisfies the boundary condition
(\ref{Kernel-summary-b}) and such that
\begin{equation}
g(\xi) = g_L(\xi) + \sum_{\ell=1}^\nu \alpha_\ell {Y}(\xi-\omega_\ell),
\label{g-form}
\end{equation}
where $\omega_\ell$ denotes the height of a plateau of $\omega$,
${Y}(\xi)$ is the ``filled'' step function, and $\alpha_\ell$ is a constant.
\end{thm}

\begin{rem}
\label{remark:singular_Casimir}
Notice that $\mathrm{Ker}(\mathcal{J}(\omega))$ is enlarged by the singular components
$\alpha_\ell {Y}(\xi-\omega_\ell)$ stemming from $\omega$ (the singular point in the phase space $H^{-1}(\Omega)$)
that has plateaus $\omega=\omega_\ell$.
Away from the singular point, the plateaus shrink to zero-measure sets in $\Omega$ and, then,
$\psi=g(\omega)$ can no longer be a member of the domain of $\mathcal{J}(\omega)$.
However, it is still a ``hyperfunction solution'' of $\{\omega,\psi\}=0$,
since $\nabla\omega$ parallels the delta-function $\nabla\psi$ at $\omega=\omega_\ell$.
The corresponding Casimir element, to be constructed in Sec.~\ref{subsec:Casimirs},
is what we call a ``singular Casimir element'' (remember the elementary example
discussed in Sec.~\ref{sec:introduction};
see also Fig.~\ref{fig:foliation}).
\end{rem}

\begin{rem}
\label{remark:J_operator-2}
In (\ref{g-form}), the function $g(\xi)$ may be chosen arbitrarily to define
an infinite number of kernel elements $\psi=g(\omega)$ satisfying (\ref{Kernel-general'}).
To find kernel elements (and in the following step of finding Casimir elements),
we solve a linear equation $\mathcal{J}(\omega) \psi=0$ (and $\mathcal{J}(\omega) \partial_\psi C(\psi)=0$)
for a given $\omega$; see Remark\,\ref{remark:J_operator}.
In the analysis of ``equilibrium points'' (Sec.\,\ref{sec:equilibrium}), however,
we will relate $\psi\in\mathrm{Ker}(\mathcal{J}(\omega))$ and $\omega$ by another
defining relation $-\Delta\psi=\omega$ so as to make $\psi$ the Clebsch potential of
$\omega$ (see (\ref{operator-K})).
Then, $g(\xi)$ is provided as a data specifying a Casimir leaf on which we seek an equilibrium point.
\end{rem}

\begin{rem}
\label{remark:more_general_solutions}
Clearly the form (\ref{g-form}) of $g(\omega)$ is rather restrictive:

(i) In the plateau region, (\ref{Kernel-general'}) has a wider class of solutions.
In fact, $\psi$ may be an arbitrary $H^1$-class function whose range may exceed
the interval $[\psi_-,\psi_+]$.  In this case, the graph of $g(\omega)$
has a ``thorn'' at $\omega_0$, and we may not integrate such a function to define a Casimir element
$G(\omega)$. (See Sec.~\ref{subsec:Casimirs}.)

(ii) In (\ref{g-form}), we restrict the continuous part $g_L(\xi)$ to be Lipschitz continuous,
by which $\psi=g(\omega)$ ($\omega\in C^{0,1}(\Omega)$) is assured of Lipschitz continuity
(thus, $\psi\in H^1(\Omega)$).
However, this condition may be weakened, depending on the specific $\omega$, i.e., it is only required that $g'(\omega)\nabla\omega\in H^1(\Omega)$.
\end{rem}

\subsection{Construction of Casimir Elements}
\label{subsec:Casimirs}

Our next mission is to ``integrate'' the kernel element
$\psi\in\mathrm{Ker}(\mathcal{J}(\omega))$ as a function of $\omega$,
and define a Casimir element $C(\omega)$,
i.e., to find a functional $C(\omega)$ such that
$\partial_{\omega} C(\omega)\in\mathrm{Ker}(\mathcal{J}(\omega))$.
To this end, the parameterized $\psi$ of (\ref{Kernel-summary}) will be used,
where the function $g(\omega)$ may have singularities as described in Theorem\,\ref{theorem:kernel2}.
The central issue of this section, then, will be to consider an appropriate
``generalized functional derivative'' by which we can define ``singular
Casimir elements'' pertinent to the singularities of the noncanonical
Poisson operator $\mathcal{J}(\omega)$.

Let us start by considering a \emph{regular} Casimir element generated by
$g(\xi)\in C(\R)$:
\begin{equation}
C_G(\omega) = \int_{\Omega} \!dx\,  G(\omega)\,,
\label{omega-Casimir}
\end{equation}
where $G(\xi) := \int\!  d\xi\, g(\xi)$.
The gradient of this functional can be readily calculated with the definition of
Sec.~\ref{subsec:gradient}:
Perturbing $\omega$ by $\epsilon\tilde{\omega}$ results in
\[
\delta C_G(\omega,\tilde{\omega})
= \epsilon\int_\Omega \!dx\,  g(\omega) \tilde{\omega} + O(\epsilon^2)\, .
\]
Hence, we obtain $\partial_\omega C_G(\omega)=g(\omega)$, proving that
$C_G(\omega)$ of (\ref{omega-Casimir}) is the Casimir element corresponding to
$g(\omega) \in \mathrm{Ker}(\mathcal{J}(\omega))$.

Now we construct a \emph{singular} Casimir element corresponding to
a general $g(\omega)$ that may have ``jumps'' at the singularity of $\mathcal{J}(\omega)$, i.e.,
 the plateaus of $\omega$.
The formal primitive function $G(\xi)$ of such a $g(\xi)$ has ``kinks''
where the classical differential does not apply
---this problem leads to the
requirement of an appropriately generalized gradient of
the functional $C_G(\omega) = \int_{\Omega}\!dx\,  G(\omega)$ generated by a kinked $G(\xi)$.

Here we invoke the \emph{Clarke gradient}\,\cite{Clarke1975}, which is a
 generalized gradient for Lipschitz-continuous functions or functionals.  Specifically, if  $F\colon\R\rightarrow\R$,  then the Clarke gradient of $F$ at $x$, denoted by $\tilde{\partial}_x F(x)$, is defined to be the convex hull of the set of limit points of the form
\begin{equation}
\lim_{j\rightarrow\infty} \partial_x F(x+\delta_j)
\quad {\rm with} \quad \lim_{j\rightarrow\infty} \delta_j = 0.
\label{Clarke1}
\end{equation}
Evidently, $\tilde{\partial}_x F(x)$ is equivalent to the classical gradient, $\partial_x F(x)$,
if $F(x)$ is continuously differentiable in the neighborhood of $x$.
Also, it is evident that a ``kink'' in $F$ yields  $\tilde{\partial}_x F(x)$ with a graph that has a ``jump'' with the gap filled as depicted in Fig.~\ref{fig:plateau}(b).
When $F(u)$ is a convex functional on a Hilbert space $X$, i.e.\ $F\colon X\rightarrow \R$, then
$\tilde{\partial}_u F(u)$ is equal to the \emph{sub-differential}:
\begin{equation}
\tilde{\partial}_u F\colon
u \mapsto \{ g\sep F(u+\delta)-F(u)\geq (g,u),~\forall \delta\in X \} ,
\label{subdifferential}
\end{equation}
which gives the \emph{maximally monotone} (i.e., the gap-filled) function \cite{Brezis}.
For the purpose of Sec.~\ref{sec:equilibrium}, a monotonic $g(\xi)=\tilde{\partial}_\xi G(\xi)$
will be sufficient.

From the above, the following conclusion is readily deducible:
\begin{cor}
\label{corollary:Casimir}
Suppose that $\omega\in C^{0,1}(\Omega)$ and $\Omega_\circ(\omega)\neq\emptyset$.
By $g(\xi)$ satisfying (\ref{Kernel-summary-b}) and (\ref{g-form}),
we define $G(\xi)$ such that $g(\xi)=\tilde{\partial}_\xi G(\xi)$.
Then, $C_G(\omega)=\int_{\Omega}\!dx\,  G(\omega)$ is a generalized Casimir element, i.e.,
$\tilde{\partial}_\omega C_G(\omega) \in \mathrm{Ker}(\mathcal{J}(\omega))$.
\end{cor}

\section{Extremal Points}
\label{sec:equilibrium}

\subsection{Extremal Points and Casimir Elements}
\label{subsec:fixed-points}
The Casimir element $C_G(\omega)$ naturally extends the set of extremal equilibrium points, the set of solutions of the dynamical
system (\ref{Vorticity_eq-2DH}) that are both equilibrium solutions and energy-Casimir extremal points.  This is done by including extremal points satisfying [cf. (\ref{stationary_points2})]
\begin{equation}
\tilde{\partial}_\omega [H(\omega) + C_G(\omega)] \ni 0,
\label{fixed_ponit}
\end{equation}
which explicitly gives
\begin{equation}
\mathcal{K}\omega \in -  g(\omega).
\label{Kernel-1-fided_point}
\end{equation}
Here we assume that $g(\xi)=\tilde{\partial}_\xi G(\xi)$ is a monotonic function, i.e., $G(\xi)$ is convex.
Then, we may define a single-valued continuous function $f= (- g)^{-1}$ and rewrite (\ref{Kernel-1-fided_point}),
denoting $\varphi=\mathcal{K}\omega$, as
\begin{equation}
-\Delta \varphi = f(\varphi).
\label{Kernel-2-fixed_point}
\end{equation}
If $f(0)\neq 0$, (\ref{Kernel-2-fixed_point}) will determine a nontrivial ($\varphi\not\equiv 0$) equilibrium
extremal point.
Notice that $f(\xi)$ (or $g(\xi)=-f^{-1}(\xi)$ or $G = \int d\xi g(\xi)$)
is a given function that specifies a Casimir leaf on which we seek an equilibrium point (see Remark \ref{remark:J_operator-2}).
Here we prove the following existence theorem:

\begin{thm}
\label{theorem:fixed-point}
There exists a finite positive number $M$, determined only by the geometry of $\Omega$,
such that if
\begin{equation}
L:= \inf_{\eta\in\R} \frac{|f(\eta)|}{|\eta|} < M,
\label{fixed-point_condition}
\end{equation}
then (\ref{Kernel-2-fixed_point}) has a solution $\varphi\in H^1_0(\Omega)$.
\end{thm}

\begin{proof}
We show the existence of the solution by Schauder's fixed-point theorem
(see for example \cite[p. 20]{Nirenberg2001} and \cite[p. 43]{bifurcation}).
First we rewrite (\ref{Kernel-2-fixed_point}) as
\begin{equation}
\omega = \mathcal{F}(\omega) := f(\mathcal{K}\omega) ,
\label{kernel-1'}
\end{equation}
where $\mathcal{K}$ is a compact operator on $L^2(\Omega)$ (cf.\  (\ref{operator-K})).
Since $f\in C(\R)$, $\mathcal{F}$ is a continuous
compact map on $L^2(\Omega)$.
We consider a closed convex subset
\[
W_\Lambda = \{ \omega \sep \| \omega \| \leq \Lambda \} ~ \subset L^2(\Omega),
\]
where $\| \omega \|$ is the norm of $L^2(\Omega)$ and the parameter $\Lambda$ will be determined later.  We
show that the compact map $\mathcal{F}$ has a fixed point in $W_\Lambda$.
By Poincar\'e's inequality, we have, for $\varphi \in H^1_0(\Omega)$,
\[
\|\varphi\| \leq c_1 \| \nabla \varphi \|,
\]
where $c_1$ is a positive number that is determined by the geometry of $\Omega$.
For $\bm{u}=\nabla\varphi\times\uep\in L^2_\sigma(\Omega)\cap H^1(\Omega)$,
we have
\[
\|\bm{u}\|  \leq c_2 \|\nabla\times\bm{u}\|,
\]
where $c_2$ is also a positive number determined by the geometry of $\Omega$.
Here $\|\bm{u}\| = \|\nabla\varphi\|$ and $\|\nabla\times\bm{u}\|= \|\Delta\varphi\|$.
Hence, denoting $c_p= c_2 c_1+c_2 +1$, we have
\[
\|\varphi\|_{H^2} := \|\varphi\| + \|\nabla\varphi\| + \|\Delta\varphi\|
\leq c_p \|\Delta\varphi\|.
\]
By Sobolev's inequality, we have
\[
\sup_\Omega |\varphi| \leq c_s \|\varphi\|_{H^2},
\]
where $c_s$ is again a positive number determined by the geometry of $\Omega$.
Summarizing these estimates, we have, for $\omega\in W_\Lambda$,
\[
\sup_\Omega |\mathcal{K}\omega| \leq c_s c_p \|\omega\| \leq c_s c_p \Lambda .
\]
For an arbitrary positive number $\delta$, there is a finite number $\eta^*$ such that
$|f(\pm\eta^*)|\leq  (L+\delta)|\eta^*|$.
If we choose $\Lambda=|\eta^*|/(c_s c_p)$, then by the monotonicity of $f(\eta)$, we find
\begin{equation}
\sup_\Omega |f(\mathcal{K}\omega)| \leq (L+\delta)|\eta^*|= (L+\delta) c_s c_p \Lambda,
\label{estimate-1}
\end{equation}
and, upon denoting $|\Omega|:=\int_\Omega \!dx$, we obtain
\[
\| f(\mathcal{K}\omega) \| \leq |\Omega|^{1/2} (L+\delta) c_s c_p \Lambda=  (L+\delta)  \Lambda/ M\,,
\]
where
\begin{equation}
M:= \frac{1}{|\Omega|^{1/2} c_s c_p}\,.
\label{estimate-M}
\end{equation}
If $(L+\delta)\leq M$, we estimate $\| f(\mathcal{K}\omega) \| \leq \Lambda$, and thus
$\mathcal{F}$ maps $W_\Lambda$ into $W_\Lambda$.  Therefore,
$\mathcal{F}$ has a fixed-point in $W_\Lambda$.
If $W_\Lambda$ contains a fixed point, then an arbitrary $W_{\Lambda'}$ with $\Lambda'>\Lambda$ contains
a fixed-point.  Since $\delta$ is arbitrary, $L<M $ is the sufficient condition
for the existence of a fixed point.
\end{proof}

\begin{cor}
\label{corollary:fixed-point}
If $L=|f(\eta^*)|/|\eta^*| >0$ for some finite $\eta^*$, then
we may assume $\delta=0$ in (\ref{estimate-1}), and the solvability condition (\ref{fixed-point_condition}) can be
modified to be $L\leq M$.
\end{cor}

The bound $M$ defined by (\ref{estimate-M}) is
related to the eigenvalue of $-\Delta=\mathcal{K}^{-1}$.
On the other hand, the constant $L$ defined by (\ref{fixed-point_condition}) is a property of the Casimir element. As noted above, because of the homogeneity of (\ref{Casmir-1}), there is at least a one-parameter family of Casimir elements and the choice of a multiplicative constant reciprocally scales $L$.  This provides an arbitrary  parameter multiplying the Casimir element in the
determining equation (\ref{fixed_ponit}), and such a parameter may be considered to be an ``eigenvalue'' characterizing the stationary point.  The no-solution example of the next section will reveal the ``nonlinear property''
of this eigenvalue problem.


\subsection{No-Solution Example}
\label{subsec:no-solution}

Here we present a class of examples that violate the
solvability condition proposed in Theorem\,\ref{theorem:fixed-point}.  This demonstrates that finding Casimir elements that generalize the extremal equation (\ref{Kernel-2-fixed_point}), does not automatically extend the set of extremal equilibria, since the resulting elliptic equation may not have a solution.  This is demonstrated by the following two propositions.

\begin{prop}
\label{proposition:nosolution}
Let $f(\varphi)$ have the form
\begin{equation}
f(\varphi) = \lambda_1 \varphi + \vartheta (\varphi),
\label{example-1}
\end{equation}
where $\lambda_1$ is the first eigenvalue of $-\Delta=\mathcal{K}^{-1}$ and $\vartheta >0$.
Then, (\ref{Kernel-2-fixed_point}) does not have a solution.
\end{prop}

\begin{proof}
We denote by $\phi_1$ the eigenfunction corresponding to the eigenvalue $\lambda_1$.
Upon taking the inner product of (\ref{Kernel-2-fixed_point}) with $\phi_1$, we obtain
\begin{equation}
(-\Delta \varphi, \phi_1) = ( \lambda_1 \varphi, \phi_1) + (\vartheta (\varphi), \phi_1).
\label{ipdel}
\end{equation}
The left-hand-side of (\ref{ipdel}) satisfies
$(-\Delta \varphi, \phi_1) = (\varphi, -\Delta\phi_1) = ( \lambda_1
\varphi, \phi_1)$, which cancels the first term on the
right-hand-side, leaving $(\vartheta (\varphi), \phi_1) =0$.
However, this is a contradiction, because $\phi_1 > 0$ (or $< 0$ if
otherwise normalized) on $\Omega$
\footnote{
Let $\Omega \subset \mathbb{R}^n$ be a bounded domain and $\lambda_k$ the eigenvalues of
$-\Delta$ with zero Dirichlet boundary condition.  It is  well-known that the eigenvalues have the order $0 < \lambda_1 < \lambda_2 \le \ldots \le
\lambda_k \le \ldots$,  with $\lambda _k \rightarrow +\infty$ as $k
\rightarrow +\infty$, and $\lambda_1$ is the unique eigenvalue with corresponding eigenfunction that does not change sign on $\Omega$, i.e., it is strictly positive or strictly negative
 on the whole of $\Omega$. Furthermore, the dimension of the eigenspace associated with $\lambda_1$ is one.}
and, by the assumption, $\vartheta (\varphi)>0$.
\end{proof}

The above result is easily generalized as follows:

\begin{prop}
\label{proposition:nosolution2}
Suppose that
\begin{equation}
f(\varphi) = \lambda_j \varphi + \vartheta (\varphi),
\label{example-2}
\end{equation}
where $\lambda_j$ is any eigenvalue of $-\Delta$ and
$0< A < \vartheta (\varphi) < B$, for some $A$ and $B$.
Let $\phi_j$ be the eigenfunction corresponding to $\lambda_j$.  Since
$\phi_j$ may not have a definite sign in $\Omega$,
we divide $\Omega$ into $\Omega_+ = \{ x\in\Omega\sep \phi_j (x)>0\}$ and
$\Omega_- = \{ x\in\Omega\sep\phi_j (x)<0\}$, and define
\[
P_+ = \int_{\Omega_+}\!\!dx\,   \phi_j \,,
\quad
P_- = -\int_{\Omega_-}\! \!dx\,  \phi_j \quad ({\rm both}\ \geq 0)
\]
If $P_+\neq P_-$, (\ref{Kernel-2-fixed_point}) does not have a solution for some ranges of
$A$ and $B$.
\end{prop}

\begin{proof}
As shown in the proof of Proposition\,\ref{proposition:nosolution2},
a solution of (\ref{Kernel-2-fixed_point}) must satisfy
\[
(\vartheta (\varphi), \phi_j) = \int_{\Omega_+}\!\!dx\,   \vartheta (\varphi) \phi_j
+ \int_{\Omega_-}\!\!dx\,   \vartheta (\varphi) \phi_j =0\,.
\]
By  assumption, we have
\begin{eqnarray*}
A \phi_j < \vartheta (\varphi) \phi_j < B \phi_j
\quad &\mathrm{in}&\Omega_+,
\\
A \phi_j > \vartheta (\varphi) \phi_j > B \phi_j
\quad &\mathrm{in}&\Omega_-.
\end{eqnarray*}
Hence, the inequalities
\[
A P_+ - B P_- < 0 < B P_+ - A P_-
\]
must hold, i.e.,
\begin{eqnarray}
A/B < P_-/P_+,
\label{contradiction1} \\
B/A > P_-/P_+ .
\label{contradiction2}
\end{eqnarray}
However, if $P_-/P_+ <1$, and $A$ and $B$ are such that $1> A/B > P_-/P_+$, then there is a contradiction with (\ref{contradiction1}).  Similarly, if
$P_-/P_+ >1$, and $A$ and $B$ are such that $1< B/A < P_-/P_+$, then there is a contradiction with  (\ref{contradiction2}).
\end{proof}


\section{Concluding Remarks}
\label{sec:conclusion}

After establishing some mathematical facts about the Hamiltonian form of the Euler equation of two-dimensional incompressible inviscid flow, we studied the center of the  Poisson algebra, i.e.,
the kernel of the noncanonical Poisson operator $\mathcal{J}(\omega)$.  Casimir elements $C(\omega)$ were obtained by ``integrating'' the ``differential equation''
\[
\mathcal{J}(\omega) \partial_\omega C(\omega) =0\,.
\]
For finite-dimensional systems with phase space coordinate $z$, this amounts to
an analysis of  $\mathcal{P}:=\mathcal{J}(z)\partial_z$, a linear partial differential operator, and
nontriviality can arise from a \emph{singularity} of $\mathcal{P}$, whence an inherent structure emerges.
Recall the simple example given in Sec.~\ref{sec:introduction} where
$\mathcal{P}=ix \partial_x$ was seen to generate the \emph{hyperfunction} Casimir $C(x)=Y(x)$.
For finite-dimensional systems  the theory naturally finds its way to
algebraic analysis: in the language of D-module theory, Casimir elements constitute
$\mathrm{Ker}(\mathcal{P}) = \mathrm{Hom}_{\mathcal{D}}(\mathrm{Coker}(\mathcal{P}),F)$,
where $\mathcal{D}$ is the ring of partial differential operators and $F$ is the
function space on which $\mathcal{P}$ operates, and
$\mathrm{Coker}(\mathcal{P})=\mathcal{D}/\mathcal{D}\mathcal{P}$ is the D-module corresponding
to the equation $\mathcal{P} C(\omega)=0$.
However, in the present study $\omega$ is a member of an infinite-dimensional Hilbert space, thus
$\mathcal{P}$ may be regarded as an infinite-dimensional
generalization of linear partial differential operators.
From the singularity of such an infinite-dimensional (or \emph{functional}) differential operator
$\mathcal{P}=\mathcal{J}(\omega) \partial_\omega$, we unearthed \emph{singular} Casimir elements, and to justify the operation of $\mathcal{P}$ on singular elements, we invoked a generalized functional
derivative (Clarke differential or sub-differential) that we denoted by $\tilde{\partial}_\omega$.

For infinite-dimensional systems, we cannot ``count'' the dimensions of
$\mathrm{Ker}(\mathcal{P})$ and $\mathrm{Ker}(\mathcal{J})$.
It is, however, evident that
$\mathrm{dim-Ker}(\mathcal{P}) < \mathrm{dim-Ker}(\mathcal{J})$, if $\mathcal{J}$ has singularities, i.e.,
singularities create ``nonintegrable'' elements of $\mathrm{Ker}(\mathcal{J})$.
As shown in Theorem\,\ref{theorem:kernel2}, a \emph{plateau}
in $\omega$ causes a singularity of $\mathcal{J}(\omega)$ and generates new
elements of $\mathrm{Ker}(\mathcal{J}(\omega))$, which can be integrated
to produce singular Casimir elements (Corollary\,\ref{corollary:Casimir}).
However, as noted in Remark \ref{remark:more_general_solutions}(i),
more general elements of $\mathrm{Ker}(\mathcal{J}(\omega))$ that are \emph{not integrable} may stem from a
 plateau singularity.
Moreover, we had to assume Lipschitz continuity for $\omega$ to obtain an
explicit relation between $\psi\in\mathrm{Ker}(\mathcal{J}(\omega))$ and $\omega$
-- otherwise we could not \emph{integrate} $\psi$ with respect to $\omega$ to construct a
Casimir element.
In the general definition of $\mathcal{J}(\omega)$, however,
$\omega$ may be nondifferentiable (we assumed only continuity),
and then, a general $\psi\in\mathrm{Ker}(\mathcal{J}(\omega))$ may not
have an integrable relation to $\omega$ (see Lemma \ref{lemma:kernel}).

In Sec.~\ref{sec:equilibrium} we solved the equation
\[
\tilde{\partial}_\omega [ H(\omega) + C(\omega) ] \ni 0
\]
for $\omega$, where the solution $\omega$ gave an equilibrium point of the dynamics
induced by a Hamiltonian $H(\omega)$.
A singular (kinked) Casimir yielded a multivalued (set-valued) gradient
$\tilde{\partial}_\omega C(\omega)$,
encompassing an infinite-dimensional solution stemming from the plateau singularity.  This arose because
in the plateau of $\omega$, $\psi\in\mathrm{Ker}(\mathcal{J}(\omega))$ is freed from
$\omega$ and may distribute arbitrarily.
The component $g_L(\omega)$ of the Casimir $C(\omega)$ of
(\ref{g-form}) represents explicitly the
regular ``dimensions" of $\mathrm{Ker}(\mathcal{J}(\omega))$.
In contrast, the undetermined dimensions pertinent to the singularity $\omega=\omega_0$
are ``implicitly'' included in the step-function component of (\ref{g-form}),
or in the kink of $C(\omega)$.
However, for a given Hamiltonian $H(\omega)$, i.e.\ a given dynamics, a specific
relation between $\varphi\in\tilde{\partial}_\omega C(\omega)$ and $\omega$
emerges.

 Theorem \ref{theorem:fixed-point} of Sec.~\ref{subsec:fixed-points} and the nonexistence examples of Sec.~\ref{subsec:no-solution} may not be new results in the theory of elliptic partial differential equations, but they do help delineate
the relationship between Hamiltonians and Casimir elements, viz.\  that Casimir elements alone do not determine the extent of the set of equilibrium points.  In the present paper, we did not discuss the bifurcation of the equilibrium points; the reader is referred to \cite{bifurcation} for a presentation of the actual state of the studies of semilinear
elliptic problems and several techniques in nonlinear analysis (see
also \cite{Nirenberg2001} and \cite{PL Lions1982}).
We also note that we have excluded nonmonotonic $g(\omega)$ that will make $f(\varphi)$ multivalued
or, more generally, equations like $\Phi(\Delta\varphi,\varphi)=0$;
cf.\ (\ref{Kernel-general''-solution}).
For fully nonlinear elliptic partial differential equations,
the reader is referred to \cite{Nirenberg2001,L Nirenberg1981,LC Evans,Caffarelli-Cabre,Caffarelli-Salazar}.

\section*{acknowledgments}
The authors acknowledge informative discussions with Yoshikazu Giga and are grateful for his suggestions.
ZY was supported by Grant-in-Aid for Scientific Research (23224014) from MEXT, Japan.
PJM was supported by U.S.~Dept.\ of Energy Contract \# DE-FG05-80ET-53088.

\section*{Appendix A. On the Poisson operator in terms of the velocity field $\bm{u}\in L^2_\sigma(\Omega)$}
\label{appendix}

Here we formulate the Euler equation, for both $n=2$ and 3, as an
evolution equation in $L_\sigma^2(\Omega)$ (see Sec.\,\ref{subsec:vorticity_equation}),
and discuss the Poisson operator $\mathcal{J}({\bm{u}})$ in this space.  This differs from the formulation of Sec.~\ref{sec:Euler_equation} in that the state variable here is the velocity field $\bm{u}\in L_\sigma^2(\Omega)$ instead
of the vorticity $\omega$.  As noted in Sec.\,\ref{subsec:vorticity_equation}, $L_\sigma^2(\Omega)$
is a closed subspace of $L^2(\Omega)$, and we have the orthogonal decomposition
(\ref{decomposition_of_L2}).
We denote by $P_\sigma$ the orthogonal projection onto $L_\sigma^2(\Omega)$.
Upon applying $P_\sigma$ to the both sides of (\ref{Euler-2}), we obtain
\begin{equation}
\partial_t \bm{u} =- P_\sigma (\nabla\times\bm{u})\times\bm{u},
\label{Euler-3}
\end{equation}
which is interpreted as the evolution equation in $L^2_\sigma(\Omega)$, where
the incompressibility condition (\ref{incompressibility}) and the boundary condition (\ref{BC})
are implied by $\bm{u}\in L^2_\sigma(\Omega)$.

For $\bm{u}\in L^2_\sigma(\Omega)$, the Hamiltonian is, of course, the kinetic energy
\begin{equation}
H(\bm{u}) = \frac{1}{2} \| \bm{u} \|^2 \,.
\label{Hamiltonian_by_u}
\end{equation}
Fixing a sufficiently smooth $\bm{u}$ acting as a parameter, viz.\  $\nabla\times\bm{u}\in C(\Omega)$,
we define for $\bm{v}\in L^2_\sigma(\Omega)$ the following
the noncanonical Poisson operator:
\begin{equation}
\mathcal{J}({\bm{u}}) \bm{v} = - P_\sigma (\nabla\times\bm{u})\times\bm{v}\, .
\label{J_by_u}
\end{equation}
As a linear operator (recall Remark \ref{remark:J_operator} of Sec.~\ref{sec:introduction}),
$\mathcal{J}(\bm{u})$ consists of the vector multiplication by $(\nabla\times\bm{u})$ followed by
projection with $P_\sigma$.
Evidently, $\mathcal{J}({\bm{u}})$ is antisymmetric:
\[
(\mathcal{J}({\bm{u}}) \bm{v}, \bm{v}')
= -(\bm{v}, \mathcal{J}({\bm{u}}) \bm{v}')
\quad \forall \bm{v}, \bm{v}'\in L_\sigma^2(\Omega).
\]
In fact, for every fixed smooth $\bm{u}$, $i\mathcal{J}(\bm{u})$ is
a self-adjoint bounded operator in $L^2_\sigma(\Omega)$.

With the $H(\bm{u})$ of (\ref{Hamiltonian_by_u}), the $\mathcal{J}(\bm{u})$ of (\ref{J_by_u}), and the gradient $\partial_{\bm{u}}$
(see Sec.\,\ref{subsec:gradient}), we may write (\ref{Euler-3}) as
\begin{equation}
\partial_t \bm{u} = \mathcal{J}({\bm{u}}) \partial_{\bm{u}} H(\bm{u}).
\label{Euler-H}
\end{equation}

Assuming $n=2$, which we do henceforth, by Lemma \ref{lemma:Clebsch_representation} we may put
$\bm{u} = \nabla\varphi\times\uep$ and $\omega = -\Delta \varphi$ with $\varphi\in H_0^1(\Omega)$.
Fixing $\omega\in C(\Omega)$ as a parameter, and
putting $\bm{v}= \nabla\psi\times\uep$ with $\psi\in H_0^1(\Omega)$,
we may write
\[
\mathcal{J}({\bm{u}}) \bm{v}
= -P_\sigma \left[\omega\uep \times(\nabla\psi\times\uep)\right]
= -P_\sigma \left[\omega \nabla\psi\right].
\]
By (\ref{decomposition_of_L2}), $\bm{v} \in \mathrm{Ker}(\mathcal{J}({\bm{u}}))$ iff
\begin{equation}
\omega\nabla\psi = \nabla \theta \quad \exists \theta\in H^1(\Omega),
\label{Kernel-u}
\end{equation}
which is equivalent to (\ref{Kernel-general'}).
Arguing just like in Sec.~\ref{subsec:kernel}, we find
solutions of (\ref{Kernel-u}) of the form
\begin{equation}
\psi = g(\omega) .
\label{Kernel-1-u}
\end{equation}

The Casimir element constructed from (\ref{Kernel-1-u}) is
\begin{equation}
C_G(\bm{u}) = \int_\Omega\!dx\,   G(\uep\cdot\nabla\times\bm{u}) = \int_\Omega\!dx\, G(\omega)\, .
\label{Casimir1}
\end{equation}
Perturbing $\bm{u}$ by $\epsilon\tilde{\bm{u}}$ and restricting $\bm{n}\cdot\tilde{\bm{u}}=0$ on $\gam $
yields $\delta\omega = \epsilon\uep\cdot\nabla\times\tilde{\bm{u}}$, and
\begin{eqnarray*}
\delta C_G(\bm{u};\tilde{\bm{u}}) &=&  
\epsilon \int_\Omega\!dx\, G'(\omega) \uep\cdot\nabla\times\tilde{\bm{u}}  + O(\epsilon^2)
\\
&=& \epsilon\int_\Omega\!dx\, \nabla G'(\omega) \times\uep\cdot\tilde{\bm{u}} + O(\epsilon^2).
\end{eqnarray*}
Hence, upon denoting $g(\omega)=G'(\omega)$ we obtain
\begin{equation}
\partial_{\bm{u}}C_G(\bm{u}) = \nabla g(\omega)\times\uep \, .
\label{Casimir2}
\end{equation}
By (\ref{Kernel-1-u}), it is evident that
$\mathcal{J}({\bm{u}})\partial_{\bm{u}}C_G(\bm{u}) =-P_\sigma [\omega\nabla g(\omega)] = 0$,
confirming that $\partial_{\bm{u}}C_G(\bm{u}) \in \mathrm{Ker}(\mathcal{J}(\bm{u}))$.

Since here, $H(\bm{u})=\|\bm{u}\|^2/2$,  we may find an equilibrium point by solving $\bm{u} = - \partial_{\bm{u}} C_G(\bm{u})$ (c.f.\   Sec.~\ref{sec:introduction}, Eq.~(\ref{stationary_points2})).





\end{document}